\documentclass[11pt,preprint]{article}

\usepackage{microtype}
\usepackage[T1]{fontenc}
\usepackage[utf8]{inputenc}
\usepackage{fullpage}
\usepackage[margin=1in]{geometry}
\usepackage{amsthm,amssymb,amsmath}  
\usepackage{xspace,enumerate}
\usepackage[dvipsnames]{xcolor}
\usepackage[colorlinks=true,urlcolor=Blue,citecolor=Green,linkcolor=BrickRed]{hyperref}
\usepackage[utf8]{inputenc}
\usepackage{thmtools}
\usepackage{thm-restate}
\usepackage{authblk}
\usepackage{todonotes}
\usepackage[noadjust]{cite}
\usepackage{subcaption}
\usepackage{framed}
\usepackage{booktabs}

\usepackage{amssymb}
\usepackage{graphicx}
\usepackage{amsthm}
\usepackage{amsmath}
\usepackage{framed}
\usepackage{enumitem}
\usepackage[mathlines]{lineno}
\usepackage[noline,noend,linesnumbered,ruled]{algorithm2e} 

\title{Faster Exponential Algorithm for Permutation Pattern Matching}
\author[ ]{Paweł Gawrychowski\thanks{Partially supported by the Bekker programme of the Polish National Agency for Academic Exchange (PPN/BEK/2020/1/00444).}}
\author[ ]{Mateusz Rzepecki}
\affil[ ]{Institute of Computer Science, University of Wrocław, Poland}
\affil[ ]{\texttt{\href{mailto:gawry@cs.uni.wroc.pl}{gawry@cs.uni.wroc.pl}}, \texttt{\href{mailto:mateuszrzepecki99@gmail.com}{mateuszrzepecki99@gmail.com}}}

\newtheorem{definition}{Definition}[section]

\newtheorem{theorem}[definition]{Theorem}
\newtheorem{lemma}[definition]{Lemma}

\newtheorem{fact}[definition]{Fact}

\newcommand{\DP}[0]{\normalfont\textbf{DP}}

\newcommand{\OO}[0]{\mathcal{O}}

\newcommand{\values}[0]{\normalfont\textbf{Values}}

\newcommand{\DD}[0]{\mathcal{D}}

\newcommand{\problem}[3]{
\begin{table}[htb]
    \centering
    \begin{tabular}{|p{0.1\linewidth} p{0.9\linewidth}|}
        \hline
        \multicolumn{2}{|l|}{\textsc{#1}} \\
        \textbf{Input:} &  #2 \\
        \textbf{Output:} & #3 \\
        \hline
    \end{tabular}
\end{table}
}

\SetKwInput{Algorithm}{Algorithm}
\SetEndCharOfAlgoLine{}
\begin{document}
\date{}
\maketitle

\thispagestyle{empty}

\begin{abstract}
The Permutation Pattern Matching problem asks, given two permutations $\sigma$ on $n$ elements and $\pi$, whether $\sigma$ admits a subsequence
with the same relative order as $\pi$ (or, in the counting version, how many such subsequences are there).
This natural problem was shown by Bose, Buss and Lubiw [IPL 1998] to be NP-complete, and hence we should seek
exact exponential time solutions.
The asymptotically fastest such solution up to date, by Berendsohn, Kozma and Marx [IPEC 2019], works in $\OO(1.6181^n)$ time.
We design a simple and faster $\OO(1.415^{n})$ time algorithm for both the detection and the counting version.
We also prove that this is optimal among a certain natural class of algorithms.
\end{abstract}

\section{Introduction}

Permutations are perhaps the most fundamental objects in discrete mathematics and computer science.
A natural notion concerning two permutations $\sigma$ and $\pi$ is that of $\sigma$ containing $\pi$ as a sub-permutation,
denoted $\pi \preceq \sigma$. 
This leads to the study of pattern avoidance, where we fix $\pi$ and want enumerate all $\sigma$ such that $\pi \preceq \sigma$
does not hold. The classical example considered by Knuth is that of $231$-avoiding permutations, which are exactly
those that can be sorted with a stack~\cite{Knuth68}.
The corresponding general algorithmic question called Permutation Pattern Matching (PPM) is to decide, for a given length-$n$
permutation $\sigma$ and length-$k$ permutation $\pi$, if $\pi \preceq \sigma$.
The counting version of this problem asks about the number of occurrences.
For instance, for $\sigma = (3, 2, 5, 4, 1)$ and $\pi = (1, 3, 2)$ there are two such occurrences, since
both $(3, 5, 4)$ and $(2, 5, 4)$ are subsequences of $\sigma$ with the same relative order as $\pi$,
and there are no other such subsequences.

\paragraph{Related work.} 
Bose, Buss and Lubiw~\cite{pub:20296} showed that PPM is NP-complete, thus we do not expect a polynomial-time
algorithm that solves this problem to exist.
Later, Jelínek and Kynčl~\cite{JelinekK17} strengthened this to hold even when $\pi$ has no decreasing subsequence
of length $3$ while $\sigma$ has no decreasing subsequence of length $4$ (which is tight in a certain sense).
However, one can still ask what is the asymptotically fastest algorithm.

For the detection version, Guillemot and Marx~\cite{DBLP:journals/corr/GuillemotM13} showed that PPM can be solved
in time $2^{\OO(k^2 \log{k})}\cdot n$. Their solution was based on a combinatorial result of Marcus and Tardos~\cite{DBLP:Marcus},
which was later refined by Fox~\cite{DBLP:mathImprovementksquared} resulting in an algorithm working in time $2^{\OO(k^2)}\cdot n$.

For the counting version, the first nontrivial algorithm working in time $\OO(n^{\frac{2}{3}k + 1})$ was designed by
Albert, Aldred, Atkinson and Holton~\cite{Albert2001AlgorithmsFP}. Ahal and Rabinovich~\cite{AhalRabinovich} soon improved this to
$n ^ {0.47 k + o(k)}$. More recently, Berendsohn, Kozma and Marx~\cite{DBLP:journals/corr/abs-1908-04673} formulated
the problem as a binary constraint satisfaction problem, and showed an algorithm working in time
$\OO(n^{k/4 + o(k)})$. On the lower bound side, they proved that solving the counting version in time
$f(k) \cdot n^{o(k/\log k)}$ would contradict the exponential time hypothesis~\cite{DBLP:journals/jcss/ImpagliazzoP01}.

In this paper we are interested in algorithms with running time depending only on the value of $n$. From such point of
view, none of the algorithms mentioned above improve upon the trivial solution working in time $2^{n}$. However,
Bruner and Lackner~\cite{Bruner2012AFA} were able to design an algorithm solving the counting version in time $\OO(1.79^n)$. 
Berendsohn, Kozma and Marx~\cite{DBLP:journals/corr/abs-1908-04673} improved on their result with an elegant
polynomial-space algorithm working in time $\OO(1.6181^n)$. In fact, their algorithm can be described in just a few sentences.
Let $f$ be one of the sought solutions, that is, a subsequence of $\sigma$ with the same relative order as $\pi$. 
We first guess $f(2), f(4), \ldots$ and then run a linear-time procedure for counting the solutions that are consistent
with the already guessed values. Thus, the main insight is that guessing how the even positions of the pattern should be
mapped to the text simplifies the problem to essentially a one-dimensional version that is easy to solve very efficiently.

\paragraph{Our results and techniques. } 
We show how to solve the counting version of PPM in time $\OO(1.415^{n})$ with a simple\footnote{See \url{https://github.com/Mateuszrze/Permutation-Pattern-Matching-Implementation} for an implementation.} linear-space algorithm.
The main insight is that instead of guessing the exact positions in $\sigma$ for some elements of $\pi$, we can
guess a range for each element, with consecutive ranges overlapping by at most one position.
Then, we run a generalised linear-time procedure
for counting the solutions that are consistent with the guesses.
Additionally, we show that our guessing strategy is close to optimal among a certain natural class of approaches.

\section{Preliminaries}\label{sec:prel}

We denote $\{1, \ldots, n\}$ by $[n]$ and $\{n, \ldots, m\}$ by $[n, m]$.
A bijective function $\sigma~:~[n]~\rightarrow~[n]$ is called a permutation of length $n$.
For any $A \subseteq [n]$, the set $\{ \sigma(a) : a \in A\}$ is denoted by $\sigma(A)$.

The input to an instance of PPM consists of a length-$n$ permutation $\sigma$ and a length-$k$ permutation $\pi$, denoted $(\sigma,\pi)$. A solution to such an instance is an injection $f : [k] \rightarrow [n]$ such that:
\begin{description}
    \item[1)] $\forall_{i, j \in [k]} \pi(i) < \pi(j) \Leftrightarrow \sigma(f(i)) < \sigma(f(j))$,
    \item[2)] $\forall_{i, j \in [k]} i < j \Leftrightarrow f(i) < f(j)$.
\end{description}
Using transitivity, the above conditions can be equivalently rewritten as follows:
\begin{description}
        \item[1)] $\forall_{i \in [k - 1]} \sigma(f(\pi^{-1}(i))) < \sigma(f(\pi^{-1}(i + 1)))$ ($Y$-axis constraints),
        \item[2)] $\forall_{i \in [k - 1]} f(i) < f(i + 1)$ ($X$-axis constraints).
\end{description}

We consider the following problem:

\problem{Permutation Pattern Matching (PPM)}{
 length-$n$ permutation $\sigma$ and length-$k$ permutation $\pi$
}{
the number of solutions to $(\sigma,\pi)$.
}
\noindent When analysing the complexity of our algorithm, we assume constant-time arithmetic operations on integers not exceeding the
number of solutions, and use the standard Word RAM model otherwise. Without such an assumption, the overall
time complexity should be multiplied by another factor of $n$.

A segment decomposition of $(\sigma,\pi)$ consists of $k$ segments $[\ell_1, r_1], \ldots, [\ell_k, r_k]$, such that:
\begin{description}
    \item[1)] $\ell_i \leq r_i$ for $i \in [k]$,
    \item[2)] $[\ell_i, r_i] \subseteq [n]$ for $i \in [k]$,
    \item[3)] $r_i \leq \ell_{i + 1}$ for $i \in [k - 1]$.
\end{description}
We stress that the segments do not need to be disjoint, however two consecutive segments overlap by at most one position.
We say that a solution $f$ to $(\sigma,\pi)$ respects such a segment decomposition when
$ f(i) \subseteq [\ell_i, r_i]$ for all $i \in [k]$.
See Figure~\ref{fig:segmentDecompositionExample} for an illustration.

\begin{figure}[hbt]
    \centering
    \includegraphics{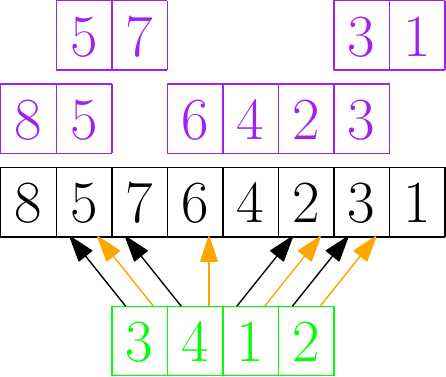}
    \caption{Black blocks represent $\sigma = (8, 5, 7, 6, 4, 2, 3, 1)$, green blocks represent $\pi = (3, 4, 1, 2)$, and purple blocks represent a segment decomposition $[1, 2], [2, 3], [4, 7], [7, 8]$ of $(\sigma,\pi)$. Black arrows represent a solution that respects the segment decomposition while orange arrows represent a solution that do not respect the segment decomposition.}
    \label{fig:segmentDecompositionExample}
\end{figure}

\section{Counting Consistent Solutions}

Let $\sigma$ be a length-$n$ permutation, $\pi$ be a length-$k$ permutation, and
consider a segment decomposition $A$ of $(\sigma,\pi)$ consisting of $[\ell_1, r_1], \ldots, [\ell_k, r_k]$.
Our goal is to count solutions to $(\sigma,\pi)$ that are consistent with $A$ with a linear-time algorithm
based on a dynamic programming.

We define $\DP_{i, j}$ as the number of functions $f: \pi^{-1}([i]) \rightarrow [n]$, such that all $Y$-axis
constraints concerning $\pi^{-1}([i])$ are fulfilled, $f(\pi^{-1}(i)) = j$, and $f(p) \subseteq [\ell_p,r_p]$
for all $p\in \pi^{-1}([i])$.
In particular, we will show that $\sum_{j=1}^{n} \DP_{k,j}$ is the total number of solutions that respect $A$.
Denoting $\values_{i} = \sigma([\ell_i, r_i])$, we see that $\DP_{i,j}$ is possibly nonzero only for $j\in \values_{i}$.
Thus, because $\sum_{i = 1}^{k} |\values_i| \leq n + k - 1$ due to the segments overlapping by at most one position,
we have only $\OO(n)$ possible nonzero values to compute and store. 
Algorithm~\ref{alg:greedY} gives a high-level description of a procedure that computes all such values.
See Figure~\ref{fig:greedyAlgorithm1} for an illustration of how it works.

\begin{algorithm}[htb]

\SetAlgoCaptionSeparator{}
\caption{}
\label{alg:greedY}
\SetKwInput{KwInput}{Input}
\SetKwInput{KwOutput}{Output}
\SetAlgoLined
\SetKwData{Left}{left}\SetKwData{This}{this}\SetKwData{Up}{up}
\tabcolsep=0pt
\begin{tabular}{@{}lp{0.77\linewidth}}
\KwInput{} & $\sigma$, $\pi$ and $A$\\
\KwOutput{} & the number of solutions to $(\sigma,\pi)$ that respect $A$
\end{tabular}
\medskip

\Begin{
\For{$i \in [k]$}{
$\values_i \leftarrow \sigma([l_i, r_i])$ \\
}
$\DP_{0,0} \leftarrow 1$ \\
\For{$i \in [k]$}{
$p \leftarrow \pi^{-1}(i)$ \\ 
\For{$j \in \values_p$}{
    $\DP_{i, j} \leftarrow \sum_{j'<j}\DP_{i-1,j'}$ \\
}
}
\Return $\sum_{j=1}^{n}\DP_{k,j}$}
\end{algorithm}

\begin{lemma}
\label{lemma:dpFormula}
$\DP_{i, j} = \sum_{j'<j}\DP_{i-1,j'}$ for $i \in [n]$ and $j \in \values_i$, where $\DP_{0, 0} = 1$ and $\DP_{0, j} = 0$ for $j> 0$.
\end{lemma}

\begin{proof}
For a fixed $i \in [n]$ and $j \in \values_i$ we know that $\DP_{i, j}$ is equal to sum over $j^{'} < j$ of the number of functions $f: \pi^{-1}([i - 1]) \rightarrow [n]$, such that all $Y$-axis constraints concerning $\pi^{-1}([i - 1])$ are fulfilled, $f(\pi^{-1}(i - 1)) = j^{'}$, and $f(p) \subseteq [\ell_p,r_p]$ for all $p\in \pi^{-1}([i - 1])$. As a result, we get that $\DP_{i, j} = \sum_{j' < j}\DP_{i - 1, j'}$. 
\end{proof}

\begin{lemma}
\label{lemma:YAndRespectImpliesX}
If a function $f : [k] \rightarrow [n]$ fulfills $Y$-axis constraints and respects $A$, then $f$ is a solution to $(\sigma, \pi)$ that respects $A$.
\end{lemma}

\begin{proof}
We know that $f$ fulfills:
\begin{description}
    \item[1)] all $Y$-axis constraints,
    \item[2)] all $X$-axis constraints, since $f$ respects $A$,  $r_i \leq \ell_{j}$ for any $i < j$ and $f$ is an injection by $Y$-axis constraints.
\end{description} 
Thus, $f$ is a solution to $(\sigma,\pi)$ that respects $A$.
\end{proof}

\begin{lemma}
\label{lemma:correctnessOfGreedy}
Algorithm~\ref{alg:greedY} computes the number of solutions to $(\sigma,\pi)$ that respect $A$.
\end{lemma}

\begin{proof}
From Lemma~\ref{lemma:dpFormula} we know that Algorithm~\ref{alg:greedY} computes $\sum_{j=1}^{n}\DP_{k,j}$. From the definition of $\DP$ we know that $\sum_{j=1}^{n}\DP_{k,j}$ is the number of functions $f: \pi^{-1}([k]) \rightarrow [n]$, such that all $Y$-axis constraints concerning $\pi^{-1}([k])$ are fulfilled, $f(\pi^{-1}(k)) \in [n]$, and $f(p) \subseteq [\ell_p,r_p]$ for all $p\in \pi^{-1}([n])$. Since $[k] = \pi^{-1}(k)$ and $f([k]) \subseteq [n]$ this simplifies to the number of functions $f: [k] \rightarrow [n]$, such that all $Y$-axis constraints are fulfilled, and $f$ respects $A$. From Lemma~\ref{lemma:YAndRespectImpliesX} we know that this is equal to the number of solutions to $(\sigma, \pi)$ that respect $A$.
\end{proof}

\begin{lemma}
\label{lemma:greedYtime}
Algorithm~\ref{alg:greedY} can be implemented to work in $\OO(n)$ time and $\OO(n)$ space.
\end{lemma}

\begin{proof} 
We start with generating and sorting the arrays $\values_1, \values_2, \ldots, \values_k$. This takes $\OO(n)$
time and space using bucket sort due to the total length of all arrays being at most $n+k-1$ as observed earlier and
all elements being from $[n]$.
Next, we calculate and store only $\DP_{i,j}$ with $j\in \values_{i}$. For each $i$, these values are stored in an array
indexed by the rank of $j$ in $\values_{i}$. The bottleneck is calculating $\sum_{j'<j}\DP_{i-1,j'}$ efficiently.
We observe that we can consider all $j\in\values_{i}$ in the increasing order. This allows us to maintain
$\sum_{j'<j}\DP_{i-1,j'}$ in total time $\OO(|\values_{i-1}|+|\values_{i}|)$ by maintaining a pointer to the
next element of $\values_{i-1}$ as we iterate over the elements of $\values_{i}$.
Summing over all $i$, the time complexity is $\sum_{i=1}^{k}\OO(|\values_{i-1}|+|\values_{i}|)=\OO(n)$.
\end{proof}

\begin{figure}[htb]
    \centering
    \begin{minipage}{0.45\textwidth}
    \includegraphics{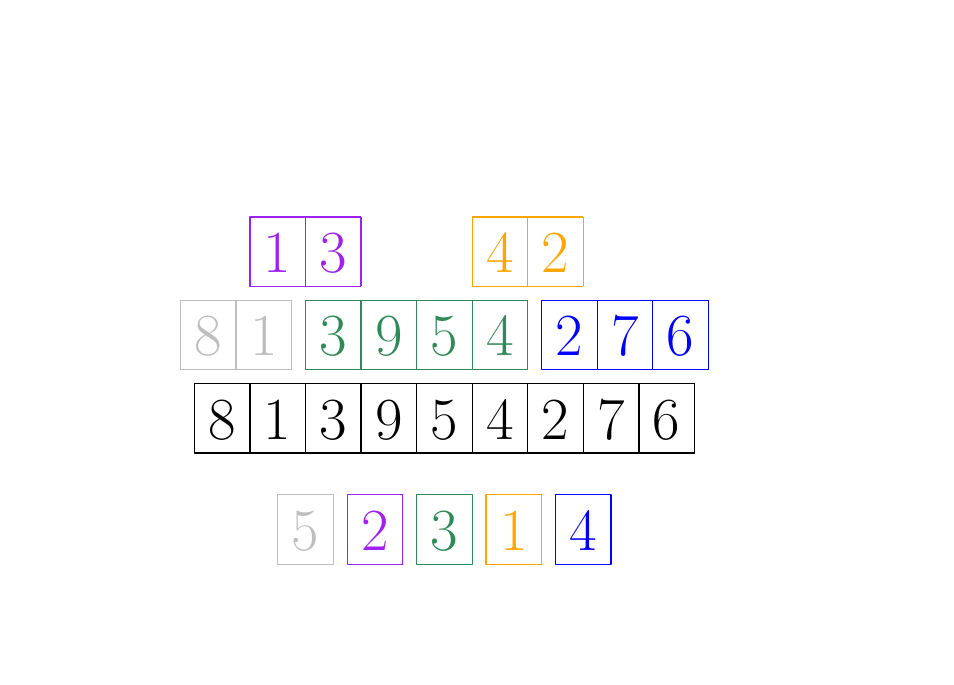}
    \end{minipage}
    \begin{minipage}{0.45\textwidth}
    \includegraphics{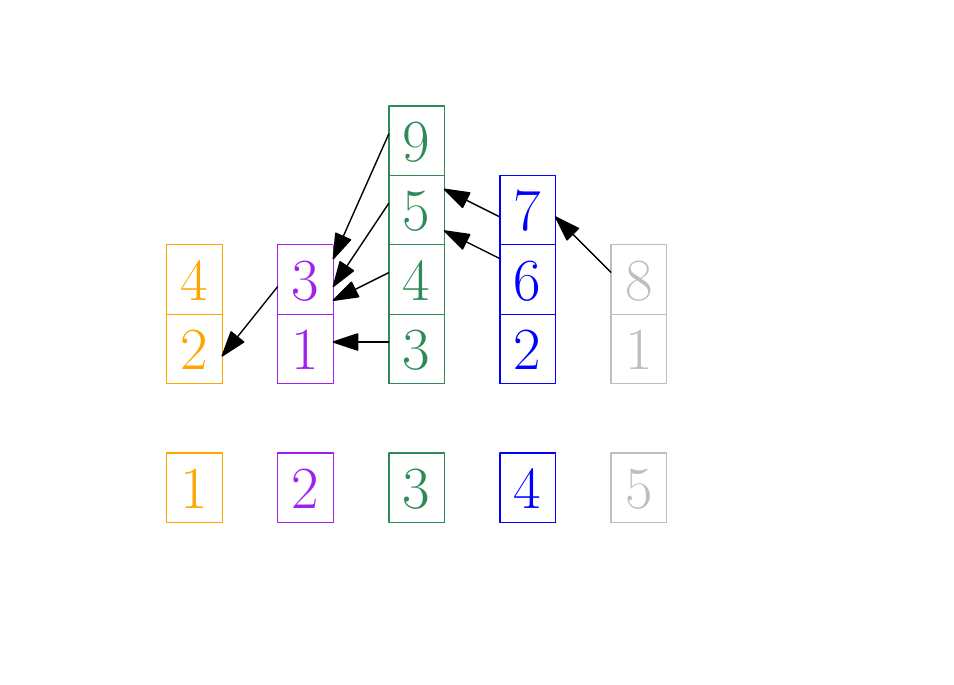}
    \end{minipage}
     \caption{Left: Black blocks represent $\sigma = (8, 1, 3, 9, 5, 4, 2, 7, 6)$, lower colored blocks represent $\pi = (5, 2, 3, 1, 4)$ and upper colored blocks represent segment decomposition $[1, 2], [2, 3], [3, 6], [6, 7], [7, 9]$ of $(\sigma,\pi)$. Right: A run of Algorithm~\ref{alg:greedY}, lower blocks represent $\pi$, upper blocks represent potential nonzero values of $\DP$, and arrows represent transitions in the algorithm.}
    \label{fig:greedyAlgorithm1}
\end{figure}

\section{Faster Algorithm}
\label{section:fastAlgorithm}

Let $\sigma$ be a length-$n$ permutation and $\pi$ be a length-$k$ permutation.
Let $A = \{i \in [k] : i \text{ is even}\} $, $B = \{i \in [n] : i \text{ is even}\}$ and $C$ be a set of all increasing functions $A \rightarrow B$. 

For each $g \in C$ let $S_g$ be a segment decomposition of $(\sigma, \pi)$ defined as $\ell_1 = 1$, $\ell_{2i} = g(2i)$,
$r_{2i} = \min \{n, g(2i) + 1\}$, $\ell_{2i + 1} = r_{2i}$, $r_{2i + 1} = \ell_{2i + 2}$, additionally if $k$ is odd then $r_k = n$.
Let $\DD = \{S_g : g \in C\}$.
Intuitively, each segment decomposition in $\DD$ fixes the value of $f(a)$, for each $a\in A$, up to xoring with $1$.
See Figure~\ref{fig:greedyAlgorithm1} (left) for an example of a segment decomposition obtained for the function $g(2) = 2$, $g(4) = 6$. 

We define $c(i)=2\lfloor i/2\rfloor$.

\begin{fact}
\label{obs:boundsOnC}
$i\in [c(i),c(i)+1]$
\end{fact}

\begin{lemma}
\label{lemma:uniqnessOfDecompositions}
For each solution $f$ to $(\sigma,\pi)$, $\DD$ contains exactly one segment decomposition respected by $f$.
\end{lemma}

\begin{proof}
Let $f$ be a solution to $(\sigma, \pi)$. Let $g \in C$ be defined by $g(2i) = c(f(2i))$. Let $S_g$ consists of $[\ell_1, r_1], [\ell_2, r_2], \ldots, [\ell_k, r_k]$. We observe that:
\begin{description}
    \item[1)] $\ell_{2i} = c(f(2i)) \leq f(2i)$, where the inequality follows from Fact~\ref{obs:boundsOnC},
    \item[2)] $r_{2i} = \min\{n, c(f(2i)) + 1\} \geq f(2i)$, where the inequality follows from Fact~\ref{obs:boundsOnC},
    \item[3)] $\ell_1 = 1 \leq f(1)$, for $i > 0$ we know that $\ell_{2i + 1} = r_{2i} = \min\{n, c(f(2i)) + 1\} \leq f(2i) + 1 \leq f(2i + 1)$, where the first inequality follows from Fact~\ref{obs:boundsOnC} while the second holds because $f$ is a solution to $(\sigma, \pi)$ and fulfills $X$-axis constraints,
    \item[4)] for $2i + 1 = k$ we know that $r_{k} = n \geq f(k)$, for $2i + 1 < k$ we know that $r_{2i + 1} = c (f(2i + 2)) \geq f(2i + 2) - 1 \geq f(2i + 1)$, where the first inequality follows from Fact~\ref{obs:boundsOnC} while the second holds because $f$ is a solution to $(\sigma, \pi)$ and fulfills $X$-axis constraints.
\end{description}
Consequently, $f(i) \in [\ell_i, r_i]$ for every $i \in [k]$, thus $f$ respects $S_g$.
See Figure~\ref{fig:fastAlgorithmUniqueness} as an illustration.

It remains to establish that $f$ does not respect any other segment decomposition $S_{g'}$. By assumption,
$g'\in C$ is a function such that there exists $i$ such that $2i \in A$ and $g(2i) \neq g'(2i)$. Let $S_{g'}$ consists of $[\ell'_1, r'_1], [\ell'_2, r'_2], \ldots, [\ell'_k, r'_k]$. 
We know that both $g(2i)$ and $g'(2i)$ are even numbers, thus $|g(2i) - g'(2i)| \geq 2$. If $g'(2i) < g(2i)$, then $r'_{2i} \leq g'(2i)  + 1 < g(2i) \leq f(2i) $ and if $g(2i) < g'(2i)$, then $f(2i) \leq g(2i) + 1 < g'(2i) = \ell'_{2i}$. In both cases $f(2i) \notin [\ell^{'}_{2i}, r'_{2i}]$, thus $f$ does not respect $S_{g'}$. \end{proof}

\begin{figure}[htb]
    \centering
    \includegraphics{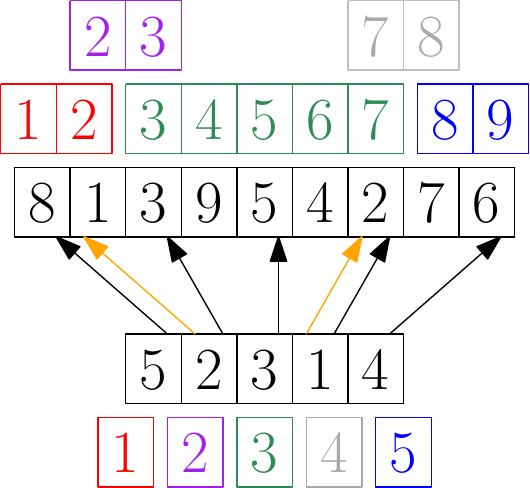}
    \caption{Situation from Lemma~\ref{lemma:uniqnessOfDecompositions} for $\sigma = (8, 1, 3, 9, 5, 4, 2, 7, 6)$, $\pi = (5, 2, 3, 1, 4)$ and a solution $f$ to $(\sigma,\pi)$, where $f(1) = 1$, $f(2) = 3$, $f(3) = 5$, $f(4) = 7$ and $f(5) = 9$. Black arrows represent $f$, orange arrows represent $g$, black blocks represent $\pi$ and $\sigma$, and colored blocks represent $S_g$.}
    \label{fig:fastAlgorithmUniqueness}
\end{figure}

\begin{theorem}
PPM can be solved in time $\OO(n \cdot 2^{\lfloor \frac{n}{2} \rfloor })=\OO(1.415^{n})$ using $\OO(n)$ space.
\end{theorem}

\begin{proof}
We iterate over all segment decompositions $D\in \DD$. For each such $D$, we run Algorithm~\ref{alg:greedY} on $\sigma,\pi, D$
to count solutions to $(\sigma,\pi)$ that respect $D$, and add this number to the result. By
Lemma~\ref{lemma:correctnessOfGreedy} and Lemma~\ref{lemma:uniqnessOfDecompositions}, this gives us the total number of
solutions to $(\sigma,\pi)$. To analyse the complexity, we observe that
the size of $\DD$ is at most $\binom{\lfloor \frac{n}{2} \rfloor}{\lfloor \frac{k}{2} \rfloor}\leq 2^{n/2}$, and we can iterate over all its elements
in total time proportional to this size and $\OO(n)$ space with a simple recursive algorithm.
For each $D\in \DD$, Algorithm~\ref{alg:greedY} runs in time $\OO(n)$ and uses only $\OO(n)$ additional space.
\end{proof}

\section{Lowerbound}

For a fixed $n$ and $k$ let $\DD_{n, k}$ be a set such that for each length-$n$ permutation $\sigma$, length-$k$ permutation $\pi$ and a solution $f$ to $(\sigma, \pi)$ there exists exactly one $D \in \DD_{n, k}$, such that $D$ is a segment decomposition of $(\sigma, \pi)$ and $f$ respects $D$.

Let $A$ be a set of all increasing functions $[\lfloor \frac{k}{2} \rfloor] \rightarrow [\lfloor \frac{n - 1}{2} \rfloor]$. For each $f \in A$ we define a function $g_f : [k] \rightarrow [n]$ by $g_f(2i) = 2 f(i)$, $g_f(2i - 1) = 2 f(i) - 1$ and if $k$ is odd then $g_f(k) = n$. Let $B = \{g_f : f \in A\}$.
By definition, $|B| = \dbinom{\lfloor \frac{n - 1}{2} \rfloor }{ \lfloor \frac{k}{2} \rfloor}$.

\begin{figure}[htb]
    \centering
    \includegraphics{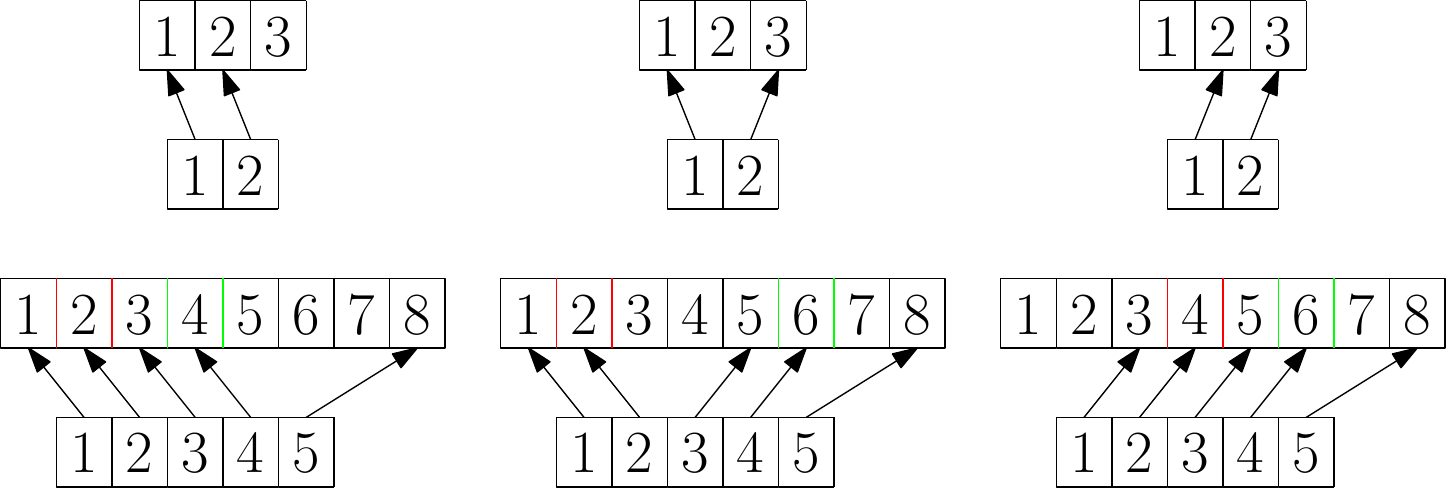}
    \caption{First row represents all functions $f \in A$ and second row represents the corresponding functions $g_f \in B$ for $n = 8$ and $k = 5$.
    Red lines show possible values of $r_1$ and green lines show possible values of $r_3$.}
    \label{fig:lowerboud}
\end{figure}

\begin{lemma}
$|\DD_{n, k}| \geq \dbinom{\lfloor \frac{n - 1}{2} \rfloor}{\lfloor \frac{k}{2} \rfloor}$
\end{lemma}

\begin{proof}
Let $f \in A$.
We can easily construct a length-$n$ permutation $\sigma$ and a length-$k$ permutation $\pi$, such that $f$ is a solution to $(\sigma, \pi)$.
Thus, by the assumed property of $\DD_{n, k}$ we know that there exists exactly one $D_f \in \DD_{n, k}$, such that $g_f$ respects $D_f$. Let $D_f$ be equal to $[\ell_1, r_1], \ldots, [\ell_k, r_k]$. We know that $2 f(i) - 1 = g_f(2i - 1) \leq r_{2i - 1} \leq l_{2i} \leq g_f(2i) = 2 f(i) $, thus $r_{2i - 1} \in \{2 f(i) - 1, 2 f(i)\}$. As a result, for each $f \in A$ we get different $D_f$, thus $|\DD_{n, k}|\geq |B| = \dbinom{\lfloor \frac{n - 1}{2} \rfloor } {\lfloor \frac{k}{2} \rfloor }$. See Figure~\ref{fig:lowerboud} as an illustration.
\end{proof}

\bibliographystyle{plainurl}
\bibliography{bibliography}

\end{document}